\documentclass[11pt]{article}

\usepackage{fullpage,amsthm}
\usepackage[T1]{fontenc}
\usepackage[utf8]{inputenc}

\usepackage{graphicx} % support the \includegraphics command and options
\usepackage{array} % for better arrays (eg matrices) in maths
\usepackage{amsmath, amssymb, amsfonts, color, verbatim}
\usepackage{epsfig,subfigure,multirow}
\usepackage[usenames,dvipsnames]{xcolor}
\usepackage{hyperref}
 \hypersetup{
  colorlinks=true,
  citecolor=ForestGreen,
}

\usepackage{url}
\usepackage{xspace}
\usepackage[mathscr]{euscript}
\usepackage{algorithm2e, algorithmic}
\usepackage{cite}
\usepackage{enumerate}
\newtheorem{theorem}{Theorem}
\newtheorem{lemma}{Lemma}[section]

\newtheorem{definition}{Definition}

\newtheorem*{claim*}{Claim}

\newcommand{\eps}{\epsilon}

\newcommand{\card}[1]{\left\vert{#1}\right\vert}

\newcommand{\set}[1]{\ensuremath{\{ #1 \}}}

\newcommand{\REM}[1]{}

\newcommand{\RomanNumbers}[1]{%
\ifcase#1 % 0
error!
\or % 1
i
\or% 2
ii
\or% 3
iii
\or% 4
iv
\or% 5
v
\else
error!
\fi
}

\newcommand\setst[2]{\ensuremath{ \{ #1 \mid} #2\ensuremath{\}}}
\def\BSet{\ensuremath{\mathscr B}}
\def\SSet{\ensuremath{\mathscr S}}

\newcommand{\edge}[1]{\ensuremath{(#1)}}
\newcommand{\match}[1]{\ensuremath{(#1)}}

\def\nSum{\ensuremath{n}}%_{sum}}}
\def\nMin{\ensuremath{n}}%_{min}}}
\newcommand{\uB}[1]{\ensuremath{u_\BSet^{#1}}}

\newcommand{\uS}[1]{\ensuremath{u_\SSet^{#1}}}

\newcommand{\val}[1]{\ensuremath{val(#1)}}

\def\market{\ensuremath{G(\BSet, \SSet, E, val)}}

\newcommand{\MSt}[1]{\ensuremath{{\mathscr S}^{#1}}}

\newcommand{\MSPPt}[1]{\ensuremath{{\mathscr S}^{#1}(P^{#1},\Pi^{#1})}}
\newcommand{\MSPP}{\MSPPt{}}

\newcommand{\FSB}[1]{\ensuremath{#1}-feasible}
\newcommand{\PSS}[1]{property $(a_#1)$ of stable states}
\newcommand{\PSSs}[2]{properties $(a_#1)$ and $(a_#2)$ of stable states}
\newcommand{\PeSS}[1]{property $(a_#1)$ of \eps-stable states}
\newcommand{\PeSSs}[2]{properties $(a_#1)$ and $(a_#2)$ of \eps-stable states}

\newcommand{\shadow}[2]{\ensuremath{\bar{#1}_{#2}}}
\def\shadowB{\shadow{S}{B}}
\def\shadowS{\shadow{B}{S}}

\def\SWV{\texttt{SW}}
\newtheorem{clm}[theorem]{Claim}
\newtheorem{cor}[theorem]{Corollary}
\newtheorem{mechanism}{Mechanism}

\numberwithin{theorem}{section}

\renewcommand{\eps}{\ensuremath{\varepsilon}}

\newcommand{\ourinfo}[1]{Department of Computer and Information Science, University of Pennsylvania. \newline Email: \texttt{#1}. 
Supported in part by National Science Foundation grants CCF-1116961, CCF-1552909, and IIS-1447470.}

\title{Fast Convergence in the Double Oral Auction} 
\author{Sepehr Assadi\thanks{\ourinfo{\{sassadi,sanjeev,yangli2\}@cis.upenn.edu}}\and 
Sanjeev Khanna\footnotemark[1] \and 
Yang Li\footnotemark[1] \and 
Rakesh Vohra\thanks{Department of Economics and the Department of Electrical and System Engineering. \newline Email: \texttt{rvohra@seas.upenn.edu}. }
}

\date{}

\begin{document}
\maketitle
\thispagestyle{empty}
\begin{abstract}
A classical trading experiment consists of a set of unit demand buyers
and unit supply sellers with identical items. Each agent's value or
opportunity cost for the item is their private information and
preferences are quasi-linear. Trade between agents employs a double
oral auction (DOA) in which both buyers and sellers call out bids or
offers which an auctioneer recognizes. Transactions resulting from
accepted bids and offers are recorded. This continues until there are
no more acceptable bids or offers.  Remarkably, the experiment
consistently terminates in a Walrasian price. The main result of this
paper is a mechanism in the spirit of the DOA that converges to a
Walrasian equilibrium in a polynomial number of steps, thus providing
a theoretical basis for the above-described empirical phenomenon. It
is well-known that computation of a Walrasian equilibrium for this
market corresponds to solving a maximum weight bipartite matching
problem. The uncoordinated but rational responses of agents thus solve
in a distributed fashion a maximum weight bipartite matching problem
that is encoded by their private valuations. We show, furthermore,
that every Walrasian equilibrium is reachable by some sequence of
responses.  This is in contrast to the well known auction algorithms
for this problem which only allow one side to make offers and thus
essentially choose an equilibrium that maximizes the surplus for the
side making offers. Our results extend to the setting where not every
agent pair is allowed to trade with each other.
\end{abstract}

\clearpage
\setcounter{page}{1}

\section{Introduction}\label{sec:intro}

Chamberlin reported on the results of a market experiment in which
prices failed to converge to a Walrasian
equilibrium~\cite{Chamber1948}.  Chamberlin's market was an instance
of the assignment model with homogeneous goods. There is a set of unit
demand buyers and a set of unit supply sellers, and all items are
identical. Each agent's value or opportunity cost for the good is
their private information and preferences are quasi-linear.
Chamberlin concluded that his results showed competitive theory to be
inadequate. Vernon Smith, in an instance of insomnia, recounted in~\cite{Smith1991} demurred:\\
``The thought occurred to me that the idea of doing an experiment was
right, but what was wrong was that if you were going to show that
competitive equilibrium was not realizable $\ldots$ you should choose
an institution of exchange that might be more favorable to yielding
competitive equilibrium. Then when such an equilibrium failed to be
approached, you would have a powerful result. This led to two ideas:
(1) why not use the double oral auction procedure, used on the stock
and commodity exchanges? (2) why not conduct the experiment in a
sequence of trading days in which supply and demand were renewed to
yield functions that were daily flows?''

Instead of Chamberlin's unstructured design, Smith used a double oral
auction (DOA) scheme in which both buyers and sellers call out bids or
offers which an auctioneer recognizes~\cite{Smith1962}. Transactions
resulting from accepted bids and offers are recorded. This continues
until there are no more acceptable bids or offers.  At the conclusion
of trading, the trades are erased, and the market reopens with
valuations and opportunity costs unchanged. The only thing that has
changed is that market participants have observed the outcomes of the
previous days trading and may adjust their expectations accordingly.
This procedure was iterated four or five times. Smith was astounded:
``I am still recovering from the shock of the experimental
results. The outcome was unbelievably consistent with competitive
price theory.''  \cite{Smith1991}(p. 156)

As noted by Daniel Friedman~\cite{Friedman1993}, the results
in~\cite{Smith1962}, replicated many times, are something of a
mystery. How is it that the agents in the DOA overcome the impediments
of both \emph{private information} and \emph{strategic uncertainty} to
arrive at the Walrasian equilibrium? A brief survey of the various
(early) theoretical attempts to do so can be found in Chapter 1
of~\cite{Friedman1993}.  Friedman concluded his survey of the
theoretical literature with a two-part conjecture.  ``First, that
competitive (Walrasian) equilibrium coincides with ordinary (complete
information) Nash Equilibrium (NE) in interesting environments for the DOA
institution. Second, that the DOA promotes some plausible sort of
learning process which eventually guides the \emph{both clever and
  not-so-clever} traders to a behavior which constitutes an `as-if'
complete-information NE.''

Over the years, the first part of Friedman's conjecture has been well
studied (see, e.g.,~\cite{Demange1986}; see also
Section~\ref{sec:ss-sw}) but the second part of the conjecture is
still left without a satisfying resolution.  The focus of this paper
is on the second part of Friedman's conjecture.  More specifically, we
design a mechanism which simulates the DOA, and prove that this
mechanism always converges to a Walrasian equilibrium in polynomially
many steps.  Our mechanism captures the following four key properties
of the DOA.

\begin{enumerate}
\item \label{prop:two-side} \emph{Two-sided market:} Agents on either
  side of the market can make actions.
\item \label{prop:private} \emph{Private information:} When making
  actions, agents have no other information besides their own
  valuations and the bids and offers submitted by others.
\item \label{prop:uncertain} \emph{Strategic uncertainty:} The agents
  have the freedom to choose their actions modulo mild rationality
  conditions.
  \item \label{prop:recognize} \emph{Arbitrary recognition:} The
    auctioneer (only) recognizes bids and offers in an
    \emph{arbitrary} order.
\end{enumerate}

Among these four properties, mechanisms that allow agents on either
side to make actions (\emph{two-sided market}) and/or limit the
information each agent has (\emph{private information}) have received
more attention in the literature (see
Section~\ref{sec:results}). However, very little is known for
mechanisms that both work for strategically uncertain agents and
recognize agents in an arbitrary order. Note that apart from resolving
the second part of Friedman's conjecture, having a mechanism with
these four properties itself is of great interest for multiple
reasons. First, in reality, the agents are typically unwilling to
share their private information to other agents or the
auctioneer. Second, agents naturally prefer to act freely as oppose to
being given a procedure and merely following it. Third, in large scale
distributed settings, it is not always possible to find a real
auctioneer who is trusted by every agent, and is capable of performing
massive computation on the data collected from all agents. In the DOA
(or in our mechanism) however, the auctioneer only recognizes actions
in an arbitrary order, which can be replaced by any standard
distributed token passing protocol, where an agent can take an action
only when he is holding the token. In other words, our mechanism
serves more like a platform (rather than a specific protocol) where
rational agents always reach a Walrasian equilibrium no matter their
actual strategy. To the best of our knowledge, no previous mechanism
enables such a `platform-like' feature. In the rest of this section,
we summarize our results and discuss in more detail the four
properties of the DOA in context of previous work.

\subsection{Our Results and Related Work}\label{sec:results}

We design a mechanism that simulates the DOA by simultaneously
capturing two-sided market, private information, strategic
uncertainty, and arbitrary recognition.  More specifically, following
the DOA, at each iteration of our mechanism, the auctioneer maintains
a list of active price submission and a tentative assignment of buyers
to sellers that `clears' the market at the current prices (note that
this can also be distributedly maintained by the agents
themselves). Among the agents who wish to make or revise an earlier
submission, an arbitrary one is recognized by the auctioneer and a new
tentative assignment is formed.  An agent can submit \emph{any} price
that strictly improves his payoff given the current submissions
(rather than being forced to make a `best' response, which is to
submit the price that maximizes payoff).  We show that as long as
agents make myopically better responses, the market always converges
to a Walrasian equilibrium in polynomial number of steps.
Furthermore, \emph{every} Walrasian equilibrium is the limit of some
sequence of better responses.  We should remark that the fact that an
agent always improves his payoff does not imply that the total payoff
of all agents always increases. For instance, a buyer can
increase his payoff by submitting a higher price and `stealing' the
current match of some other buyer (whose payoff would drop).

To the best of our knowledge, no existing mechanism captures all four
properties for the DOA that we proposed in this paper.  For most of
the early work on auction based algorithms
(e.g.,~\cite{ShapleyS1971,Kuhn10,Crawford1981,Demange1986,Bertsekas1979}),
unlike the DOA, only one side of the market can make offers. By
permitting only one side of the market to make offers, the auction
methods essentially pick the Walrasian equilibrium (equilibria are not
unique) that maximizes the total surplus of the side making the
offers.

For two-sided auction based
algorithms~\cite{BertsekasC1992,LinearBook}, along with the `learning'
based algorithms studied more recently~\cite{nax2013,kanoria2011},
agents are required to follow a specific algorithm (or protocol) that
determines their actions (and hence violates strategic
uncertainty). For example, \cite{BertsekasC1992} requires that when an
agent is activated, a buyer always matches to the `best' seller and a
seller always matches to the `best' buyer (i.e., agents only make
myopically \emph{best} responses, which is not the case for the DOA).
\cite{kanoria2011} has agents submit bids based on their current best
alternative offer and prices are updated according to a common formula
relying on knowledge of the agents opportunity costs and marginal
values.  \cite{nax2013}, though not requiring agents to always make
myopically best responses, has agents follow a specific (randomized)
algorithm to submit conditional bids and choose matches.  We should
emphasize that agents acting based on some \emph{random} process is
different from agents being strategically uncertain. In particular, for
the participants of the original DOA experiment (of~\cite{Smith1962}),
there is no a priori reason to believe that they were following some
specific random procedure during the experiment. On the contrary, as
stated in Friedman's conjecture, there are \emph{clever and
  not-so-clever} participants, and hence different agents could have
completely different strategies and their strategies might even change
when, for instance, seeing more agents matching with each other, or by
observing the strategies of other agents.  Therefore, analyzing a
process where agents are strategically uncertain can be distinctly
more complex than analyzing the case where agents behave in accordance
with a well-defined stochastic process.  In this paper, we consider an
extremely general model of the agents: the agents are acting
arbitrarily while only following some mild rationality conditions.
Indeed, proving fast convergence (or even just convergence) for a
mechanism with agents that are strategically uncertain is one of the
main challenges of this work.

Arbitrary recognition is another critical challenge for designing our
mechanism.  For example, the work of~\cite{Pradelski2015,nax2013}
deploys randomization in the process of recognizing agents. This is
again in contrast to the original DOA experiment, since the auctioneer
did not use a randomized procedure when recognizing actions, and it is
unlikely that the participants \emph{decide} to make an action
following some random process (in fact, some participants might be
more `active' than others, which could lead to the `quieter'
participants barely getting \emph{any} chance to make actions, as long
as the `active' agents are still making actions).

The classical work on the \emph{stable matching}
problem~\cite{Gale1962stable-matching} serves as a very good
illustration for the importance of arbitrary recognition.
Knuth~\cite{knuth1976} proposed the following algorithm for finding a
stable matching. Start with an arbitrary matching; if it is stable,
stop; otherwise, pick a blocking pair and match them; repeat this
process until a stable matching is found.  Knuth showed that the
algorithm could cycle if the blocking pair is picked
\emph{arbitrarily}. Later, \cite{Rrandom} showed that picking the
blocking pairs at random suffices to ensure that the algorithm
eventually converges to a stable matching, which suggests that it is
the arbitrary selection of blocking pairs that causes Knuth's
algorithm to not converge.

The setting of Knuth's algorithm is very similar to the process of the
DOA in the sense that in any step of the DOA, a temporary matching is
maintained and agents can make actions to (possibly) change the
current matching. But perhaps surprisingly, we show that arbitrary
recognition does not cause the DOA to suffer from the same cycling
problem as Knuth's algorithm. The main reason, or the main difference
between the two models is that our assignment model involves both
matching and prices, while Knuth's algorithm only involves
matchings. As a consequence, in our mechanism, the preferences of the
agents change over time (since an agent always favors the better price
submission, the preferences could change when new prices are
submitted). In the instance that leads Knuth's algorithm to cycle
(see~\cite{knuth1976}), the fundamental cause is that the preferences
of \emph{all} agents form a cycle.  However, in our mechanism,
preferences (though changing) are always consistent for all agents.

Based on this observation, we establish the limit of the DOA by
introducing a small friction into the market: restricting the set of
agents on the other side that each agent can trade with\footnote{In
  Chamberlin's experiment, buyers and sellers had to seek each other
  out to determine prices. This search cost meant that each agent was
  not necessarily aware of all prices on the other side of the
  market.}. We show that in this case, there is an instance with a
specific adversarial order of recognizing agents such that following
this order, the preferences of the agents (over the entire order) form
a cycle and the DOA may never converge.  Finally, we complete the
story by showing that if we change the mechanism to recognize agents
randomly, with high probability, a Walrasian equilibrium will be
reached in polynomial number of steps.  This further emphasizes the
distinction between random recognition and arbitrary recognition for
DOA-like mechanisms.

\medskip
\noindent
{\bf Organization:} The rest of the paper is organized as follows. In
Section~\ref{sec:preliminary}, we formally introduce the model of the
market and develop some concepts and notation used throughout the
paper. Section~\ref{sec:ss-sw} establishes a connection between the
stable states of the market and social welfare.
Our main results are presented in Section~\ref{sec:cvg}. We describe
our DOA style mechanism and show that in markets with no trading
restrictions, it converges in a number of steps that is polynomially
bounded in the number of agents. We then show that when each agent is
restricted to trade only with an arbitrary subset of agents on the
other side, the mechanism need not converge. A randomized variant of
our mechanism is then presented to overcome this issue.  Finally, we
conclude with some directions for future work in
Section~\ref{sec:conclusion}.

\section{Preliminaries}
\label{sec:preliminary}

We will use the terms `player' and `agent' interchangeably throughout
the paper. We use $B$ to represent a buyer, $S$ for a seller, and $Z$
for either of them. Also, $b$ is used as the bid submitted by a buyer
and $s$ as the offer from a seller.

\begin{definition}[Market]
A market is denoted by \market, where \BSet\ and \SSet\ are the sets of
buyers and sellers, respectively. Each buyer $B \in \BSet$ is endowed
with a valuation of the item, and each seller $S \in \SSet$ has an
opportunity cost for the item. We slightly abuse the terminology and
refer to both of these values as the valuation of the agent for the
item. The valuation of any agent $Z$ is chosen from range $[0,1]$, and
denoted by $\val{Z}$.  Finally, $E$ is the set of undirected edges
between \BSet\ and \SSet, which determines the buyer-seller pairs that
may trade.
\end{definition}

Let $m = \card{E}$ and $n = \card{\BSet} + \card{\SSet}$.

\begin{definition}[Market State]
The \emph{state} of a market at time $t$ is denoted 
\MSPPt{t}\ (\MSPP\ for short, if time is clear or not relevant), where
$P$ is a \emph{price function} revealing the price submission of each
player and $\Pi$ is a \emph{matching} between \BSet\ and \SSet,
indicating which players are currently paired.  In other words, the
bid (offer) of a buyer $B$ (seller $S$) is $P(B)$ ($P(S)$), and $B$,
$S$ are paired in $\Pi$ iff $\match{B,S} \in \Pi$.  In addition, we
denote a player $Z \in \Pi$ iff $Z$ is matched with some other player
in $\Pi$, and denote his match by $\Pi(Z)$.

Furthermore, the state where each buyer submits a bid of $0$, each
seller submits an offer of $1$, and no player is matched is called the
\emph{zero-information state}.
\end{definition}

We use the term zero-information because no player reveals 
non-trivial information about his valuation in this state.

\begin{definition}[Valid State]
A state is called \emph{valid} iff $(a_1)$ the price submitted by each
buyer (seller) is lower (higher) than his valuation, $(a_2)$ two
players are matched only when there is an edge between them, and
$(a_3)$ for any pair in the matching, the bid of the buyer is no
smaller than the offer of the seller.
\end{definition}

In the following, we restrict attention to states that are valid.

\begin{definition}[Utility]
\label{def:utl}
For a market \market at state \MSPP, the \emph{utility} of a buyer is
defined as $\val{B} - P(B)$, if $B$ receives an item, and zero
otherwise. Similarly, the utility of a seller is defined as $P(S) -
\val{S}$, if $S$ trades his item, and zero otherwise.
\end{definition}

Note that what we have called utility is also called surplus.

\begin{definition}[Stable State]
\label{def:NE}
A \emph{stable state} of a market \market\ is a state
\MSPP\ s.t. $(a_1)$ for all $\edge{B,S} \in E$, $P(B) \le P(S)$ $(a_2)$
if $Z \notin \Pi$, then $P(Z) = \val{Z}$, and $(a_3)$ if $\match{B,S} \in
\Pi$, then $P(B) = P(S)$. 
\end{definition}

Suppose $\MSPP$ is not stable. Then, one of the following must be true.
\begin{enumerate}
\item There exists $(B,S) \in E$ such that $P(B) > P(S)$. Then, both
  $B$ and $S$ could strictly increase their utility by trading with
  each other using the average of their prices.
\item There exists $Z \not \in \Pi$ such that $P(Z) \neq val(Z)$. This
  agent could raise his bid (if a buyer) or lower his offer (if a
  seller), without reducing his utility and having a better
  opportunity to trade. 
\item There exists $(B,S) \in \Pi$ such that $P(B) > P(S)$ ($P(B) <
  P(S)$ results in an invalid state). One of the agents could do
  better by either raising his offer or lowering his bid.
\end{enumerate}

\begin{definition}[$\eps$-Stable State]
  For any $\eps \geq 0$, a state \MSPP\ of a market \market\ is
  $\eps$-stable iff $(a_1)$ for any $\edge{B,S} \in E$, $P(B) -
  P(S) \le \eps$ $(a_2)$ if player $Z \notin \Pi$, $P(Z) = \val{Z}$,
  and $(a_3)$ if $\match{B,S} \in \Pi$, $P(B) = P(S)$.
\end{definition}
Note that the only difference between a stable state and an
$\eps$-stable state lies in the first property. At any $\eps$-stable
state, no matched player will have a move to increase his utility by
more than $\eps$.

\begin{definition}[Social Welfare]
For a market \market\ with a matching $\Pi$, the social welfare (SW)
of this matching is defined as the sum of the valuation of the matched
buyers minus the total opportunity cost of the matched sellers.  We
denote by $\SWV_\Pi$ the SW of matching $\Pi$. 
\end{definition}

\begin{definition}[$\eps$-approximate SW]
For any market, a matching $\Pi$ is said to give an $\eps$-approximate
SW if $\SWV_\Pi \ge \SWV_{\Pi^*} - \nMin \eps$ for any $\Pi^*$ that
maximizes SW. In other words, on average, the social welfare collected
from each player using $\Pi$ is at most $\eps$ less than that
collected using $\Pi^*$.
\end{definition}

\section{Stable State and Social Welfare}
\label{sec:ss-sw}
In this section we mainly establish the connection between stables states and social
welfare in the market. We emphasize that most results in this section are well known in
the literature and stated here for the sake of completeness.

The problem of finding a matching that maximizes SW can be formulated
as a linear program (LP) (see~\cite{LinearBook} for example). For any
edge $\edge{B,S} \in E$, let $x_{B,S}$ be the variable indicating
whether \edge{B,S}\ is selected in the matching or not, and define
weight of the edge, $w_{B,S} = \val{B} - \val{S}$. Therefore, the LP
(primal) and its dual can be defined as follows.
\begin{align*}
&\max{\sum_{\edge{B,S}\in E}  w_{B,S} \cdot x_{B,S}}
  & &\min{\sum_{ B\in \BSet}  y_{B} + \sum_{S\in \SSet}  y_{S}}\\
\text{s.t. } &\forall B^* \in \BSet, \sum_{\edge{B^*,S} \in E}x_{B^*,S} \le 1
 & \text{s.t. } &\forall \edge{B,S} \in E, y_B + y_S \ge w_{B,S}\\
&\forall S^* \in \SSet, \sum_{\edge{B,S^*} \in E}x_{B,S^*} \le 1
& & y_{B}, y_{S} \ge 0\\
& x_{B,S} \ge 0
\end{align*}
In the following, we will refer to the above linear programs as
`primal' and `dual', respectively.  The dual variables $y$ can be
interpreted as the utilities that agents enjoy assuming every buyer
gets an item and every seller sells the item. Since it only depends on
the price function, we call this price-wise utility.  The constraint
$y_B + y_S \ge w_{B,S}$ essentially states that the sum of the
utilities obtained by $(B,S)$ must be at least as large as their gains
from trade.

We use $x^\Pi$ to denote the characteristic function of matching
$\Pi$, i.e., $x^\Pi _{B,S} = 1$ iff $\match{B,S} \in \Pi$, and use
$y^P$ to denote the price-wise utility function of a price function
$P$, i.e., $y^P_B = \val{B} - P(B)$ and $y^P_S = P(S)- \val{S}$.  It
is well known that SW is maximized at a Walrasian equilibrium
(see~\cite{LinearBook}) and we state here a similar result for stable
states (a simple proof can be found in Appendix~\ref{sec:ss-msw}).

\begin{theorem}\label{thm:stable-pd}
A state \MSPP\ is stable iff $x^\Pi$ is an optimal solution for
the primal and $y^P$ is an optimal solution for the dual.
\end{theorem}

Theorem~\ref{thm:stable-pd} states that any stable state maximizes SW.
In other words, a stable state is a Walrasian equilibrium of the
market. Moreover, \emph{any} pair of optimal primal and dual solutions
can form a stable state.
We now show that for a sufficiently small $\eps$, an $\eps$-stable state
is almost as good as stable states in terms of achieving maximum SW.
We defer the proof of the following theorem to Appendix~\ref{sec:proof-esw}.

\begin{theorem}\label{thm:MSW-eSS}
For any market \market, for any $\eps > 0$, any $\eps$-stable state
realizes an $\eps$-approximate SW. Moreover, if we define $\delta =
\min \set{|val(Z_1) - val(Z_2)| \mid Z_1,Z_2 \in \BSet \cup \SSet ,
  val(Z_1) \neq val(Z_2)}$, then for $0 \le \eps <\delta/\nMin$,
any $\eps$-stable state maximizes SW.
\end{theorem}

We note that~\cite{Demange1986} also shows that a $\eps$-stable state
realizes an $\eps$-approximate SW. However, the bound on $\eps$ given
in Theorem~\ref{thm:MSW-eSS} is new.  In~\cite{LinearBook}, using
$\eps$-complementary slackness, Bertsekas shows that for integer
valuations, any $\eps$-stable state achieves maximum SW if $\eps <
1/n$. Therefore, for fractional valuations, by scaling valuations with
a suitably large factor $L$, one can make the valuations integers, and
deduce that $\eps < 1/(n L)$ suffices for achieving maximum SW.  Note
that $L$ is at least $1/\delta$ but can possibly be much larger.

We should point out that the bound $\eps < \delta/\nMin$ is not an
immediate consequence of the fact that any matching in an
$\eps$-stable state is an $\eps$-approximate SW, by arguing that the
smallest non-zero difference in SW of two matchings is at least
$\delta$. Consider a market whose trading graph is a complete
bipartite graph, with four players, where $\val{B_1} = 0.1$,
$\val{S_1} = 0.05$ $\val{B_2} = 0.2001$, $\val{S_2} = 0.15$. The
difference of valuation price between any two players is lower bounded
by $0.05$ ($\delta = 0.05$) but $B_1, S_1$ yields a SW of $0.05$ and
$B_2, S_2$ yields a SW of $0.0501$ and the difference in SW could be
made arbitrarily small.

Finally, it is worth mentioning that the fact that $\eps$-stable state
gives $\eps$-approximate SW does have a corollary as follows, which is
a weaker result compared to Theorem~\ref{thm:MSW-eSS}: If for any
$\edge{B,S} \in E$, $\val{B} - \val{S}$ is an integer multiple of
$\delta$, then for any $0 \le \eps < \delta/\nMin$, an $\eps$-stable
state always maximizes SW.

\section{Convergence to a Stable State}
\label{sec:cvg}

We establish our main results in this section.  We will start by
describing a mechanism in the spirit of DOA, and show that for any
\emph{well-behaved} stable state, there is a sequence of agent moves
that leads to this state.  When the trading graph is a complete
bipartite graph, i.e, the case of the DOA expriment, we show that convergence to a stable state occurs in
number of steps that is polynomially bounded in the number
agents. However, convergence to a stable state is not
guaranteed when the trading graph is an incomplete bipartite graph.  We
propose a natural randomized extension of our mechanism, and show that
with high probability, the market will converge to a stable state in
number of steps that is polynomially bounded in the number of agents.

\subsection{The Main Mechanism}
\label{subsec:main_mech}

To describe our mechanism, we need the notion of an
\emph{$\eps$-interested} player.

\begin{definition}[$\eps$-Interested Player]
\label{def:interested}
For a market at state \MSPP~with any parameter $\eps > 0$, a seller
$S$ is said to be \emph{\eps-interested} in his neighbor $B$ iff
either $(a)$ $P(B) \ge P(S)$ and $S \notin \Pi$, or $(b)$ $P(B) - P(S) \ge
\eps$ and $S \in \Pi$. The set of buyers interested in a seller
$S$ is defined analogously.
\end{definition}

When the parameter $\eps$ is clear from the context, we will simply
refer to an \emph{\eps-interested} player as an interested player.

\begin{mechanism}\label{mech:overbid}
(with input parameter $\eps > 0$)
\begin{itemize}
\item \textbf{Activity Rule:} Among the unmatched buyers, any buyer
  that neither submits a new higher bid nor has a seller that is
  interested in him, is labeled as inactive. All other unmatched
  buyers are labeled as active. An active (inactive) seller is defined
  analogously. An inactive player changes his status iff some player
  on the other side matches with him.\footnote{This is common for
    eliminating no trade equilibria.}
\item \textbf{Minimum Increment:} 
  Each submitted price must be an integer multiple of $\eps$.
  \footnote{This is part of many experimental implementations of the
    DOA.}
\item \textbf{Recognition:} Among all active players, an arbitrary one
  is recognized.
\item \textbf{Matching:} After a buyer $B$ is recognized, $B$ will
  choose an interested seller to match with if one exists. If the
  offer of the seller is lower than the bid $b$, it is immediately
  raised to $b$. The seller action is defined analogously.
\item \textbf{Tie Breaking:} When choosing a player on the other side
  to match to, an unmatched player is given priority (the {\em
    unmatched first} rule).
\end{itemize}
\end{mechanism}

In each iteration, players are partitioned into two sets based on
whether they are matched or not.  The unmatched players are further
partitioned into active players and inactive players. The only players
with a myopic incentive to revise their submissions are those that are
not matched.

Observe that since a buyer will never submit a bid higher than his
valuation, and a seller will never make an offer below his own
opportunity cost, by submitting only prices that are integer multiples
of $\eps$, an agent might not be able to submit his true valuation.
However, since an agent can always submit a price at most $\eps$ away
from the true valuation, if we pretend that the `close to valuation'
prices are true valuations, the maximum SW will decrease by at most
$n\eps$. By picking $\eps'=\eps/2$, if the market converges to an
$\eps'$-stable state, we still guarantee that the SW of the final
state is at most $n\eps$ away from the maximum SW.

When a buyer $B$ chooses to increase his current bid: if $s$ denotes
the lowest offer in the neighborhood of $B$, and $s'$ denotes the
lowest offer of any unmatched seller in the neighborhood of $B$, then
the new bid of $B$ can be at most $\min\set{s + \eps, s'}$. We refer
to this as the {\em increment} rule. This may be viewed as a
consequence of rationality -- there is no incentive for a buyer to bid
above the price needed to make a deal with some seller. A similar rule
applies to sellers. With a slight abuse of the terminology, we call
either rules increment rule. Notice, a player indifferent between
submitting a new price and keeping his price unchanged will be assumed
to break ties in favor of activity.

Note that the role of the auctioneer in Mechanism~(\ref{mech:overbid})
is restricted to recognize agent actions, but never select actions for
agents. In fact, the existence of an auctioneer is not even necessary
for the mechanism to work. Minimum increment can be interpreted as
setting the currency of the market to be $\eps$. Arbitrary
recognition can be achieved by a first come, first served
principle. Activity rule and matching are both designed to ensure that
players will keep making actions (submitting a new price or forming a
valid match) if one exists.

We first prove some properties of Mechanism (\ref{mech:overbid}).

\begin{clm}\label{clm:always-estable-1-3}
For any market, if we use Mechanism (\ref{mech:overbid}) with any
input parameter $\eps > 0$ and start from any state that satisfies
\PeSSs{1}{3}, any state reached satisfies \PeSSs{1}{3}.
\end{clm}

By the increment rule and matching rule respectively, the reached
state satisfies \PeSSs{1}{3}.

If a state \MSPP~satisfies $\forall \edge{B,S} \in E, P(B) \le P(S)$,
then we call it a {\em valid starting state}. Note that a valid
starting state satisfies \PeSSs{1}{3} (a valid $\Pi$ matches a buyer
to a seller only if the bid price of the buyer is at least the offer
price of the seller). In the following, we only consider markets that
begin with a valid starting state, and hence a matched player will
never have a move to increase his utility by more than $\eps$.

\begin{clm}
For any market, if we use Mechanism (\ref{mech:overbid}) and begin
with a valid starting state, then any final state of the market is
$\eps$-stable.
\end{clm}

Claim~\ref{clm:always-estable-1-3} ensures that the final state
satisfies \PeSSs{1}{3}, and \PeSS{2} holds because an unmatched buyer
will always submit a new higher bid to avoid being inactive, unless he
reaches his valuation. Same for the unmatched sellers.

Note that by Theorem~\ref{thm:MSW-eSS}, if a market converges to an
\eps-stable state, it always realizes \eps-approximate SW.

\begin{definition}[Well-behaved]
A stable state \MSPP, is \emph{well-behaved} iff $(a_1)$ for any
$\edge{B,S} \in E$, if $B \notin \Pi$ and $S\notin \Pi$, then $P(B) <
P(S)$.  An \eps-stable state \MSPP, is \emph{well-behaved} iff not
only property $(a_1)$ is satisfied but also $(a_2)$ for any
$\edge{B,S} \in E$, if either $B \notin \Pi$ or $S \notin \Pi$, then
$P(B) \leq P(S)$.
\end{definition}
Note that the states ruled out by properties $(a_1)$ and $(a_2)$ of
well-behaved states are the corner cases where a buyer-seller pair
having the same valuation (thus having no contribution to SW) are not
chosen in the matching, or players who can obtain utility at most
$\eps$ stop attempting to match with others.

\begin{theorem}
\label{lma:mec2-reach}
For any $\eps > 0$, if we use Mechanism~(\ref{mech:overbid}), and
start from the zero-information state, any well-behaved \eps-stable
state can be reached via a sequence of at most $\nSum$ moves. Hence,
any well-behaved stable state is also reachable.
\end{theorem}

\begin{proof}
Given an \eps-stable state \MSPP, sort all pairs in $\Pi$ in
decreasing order of prices (arbitrarily break the ties), and denote
the ordering as $O$. We propose a two-stage procedure: first stage
handles the players in the matching and second stage deals with the
remainders. Note that we only need to justify that the increment rule
and unmatched first rule hold for every action.

In stage one, choose pairs of players following the order defined by
$O$.  For each pair $\match{B,S}$, let the buyer submit $P(B)$, and
then, let the seller submit $P(S)=P(B)$ and match with $B$.
When $B$ submits $P(B)$, no seller is submitting a price lower then
$P(B)$, hence the increment rule is satisfied.  The unmatched sellers
are submitting $1$, and hence either no one is interested in $B$ or
all of them are interested in $B$ (if $P(B) = 1$, i.e., $P(B)$ is no less than
the seller prices). In the later case, $B$ can directly
match with $S$.

For $S$, assume the highest bid he can see in his neighborhood is
$P(B')$ submitted by buyer $B'$.  By \PeSS{1}, $P(B') \le P(S) + \eps
= P(B) + \eps$. Among the unmatched neighbors of $S$, $B$ is the one
submitting the highest price, and $P(S) = P(B) \ge \max\set{P(B') -
  \eps, P(B)}$, the increment rule is satisfied. Since $S$ matches
with unmatched buyer $B$, the unmatched first rule is also satisfied.

In stage two, choose the unmatched players with an arbitrary order and
let them submit their valuations. For any unmatched buyer $B$, by
property $(a_2)$ of well-behaved states, $P(S) \ge P(B)$ for any
seller $S$ visible to $B$, hence the increment rule is satisfied. In
addition, for any unmatched seller $S$, by property $(a_1)$ of
well-behaved states, $P(B) <P(S)$, thus $B$ cannot match with $S$.  By
analogy, any unmatched seller will also make a valid move and remain
unmatched.

Thus, after exactly \nSum~steps the two stages end, and the market is
in state \MSPP. 
\end{proof}

\subsection{Complete Bipartite Graphs}
\label{sec:cvg-complete}
We now prove that market with complete bipartite trading graph will
always converge when using Mechanism~(\ref{mech:overbid}).

\begin{theorem}
\label{thm:cvg-complete-2}
For a market whose trading graph is a complete bipartite graph, if we
use Mechanism (\ref{mech:overbid}) with any input parameter $\eps >
0$, and begin with any valid starting state, then the market will
converge to an \eps-stable state after at most $n^3/\eps$ steps.
\end{theorem}

We need the following lemma to prove Theorem~\ref{thm:cvg-complete-2}.

\begin{lemma}
\label{lma:no-over-price}
For a market \market\ whose trading graph is a complete bipartite
graph, if we use Mechanism (\ref{mech:overbid}) with any input
parameter $\eps > 0$, then at any state \MSPP\ reached from a valid
starting state, for any $\edge{B,S} \in E$, if $P(B) > P(S)$, then
both $B$ and $S$ are matched.
\end{lemma}
\begin{proof}
Assume by contradiction that there exists some $\edge{B,S} \in E$ with
$P(B) > P(S)$ and, wlog, $B$ being unmatched. Since in the starting
state, $P(B) \le P(S)$, let $t$ be the first time that this
happens. Therefore, at time $t-1$, either $P(B) \le P(S)$ or $B$ is
matched. Note that since the prices are integer multiples of \eps, a
state with $P(B) > P(S)$ implies $P(B) - P(S) \ge \eps$. On the other
hand, since \PeSS{1} always holds, $P(B) - P(S) \le \eps$. Thus $P(B)
= P(S) + \eps$ at time $t$.

If $P(B) \le P(S)$ at time $t-1$, $P(B) > P(S)$ can only be a
consequence of either $B$ or $S$ being recognized. If $B$ is
recognized and submits a bid of $P(S) + \eps$, since $S$ is interested
in $B$, by the matching rule, $B$ will be matched. If $S$ is
recognized and submits an offer of $P(B)-\eps$, by the increment rule,
$B$ must be matched (otherwise $S$ would not submit an offer lower
than $P(B)$), a contradiction.

Assume that $B$ was matched to some seller $S'$ at time $t-1$.  The
only valid action at time $t-1$ that can make $B$ unmatched is if some
buyer $B'$ overbids $B$ and match with seller $S'$. If $S = S'$, then
after the move, $P(B) < P(S)$, a contradiction. If $S \neq S'$, then
this move will not change the bid of $B$ or offer price of $S$, and
hence, $P(B) = P(S) + \eps$ in time $t-1$. Since the trading graph is
a complete bipartite graph, $S$ is a neighbor of $B'$. By the
increment rule, $B'$ can only submit a price at most equal to $P(S) +
\eps =P(B)$, thus $B'$ is unable to overbid $B$, a contradiction.
\end{proof}

\begin{definition}[$\gamma$-feasible]
A market state \MSPP\ is said to be $\gamma$-feasible iff there are
exactly $\gamma$ matches in $\Pi$.
\end{definition}

\begin{proof}[Proof of Theorem~\ref{thm:cvg-complete-2}]
Assume at any time $t$, the state \MSt{t} of the market is
\FSB{\gamma^t}. Define the following potential function
\[
   \Phi_{P} = \sum_{S_i \in \SSet} P(S_i) + \sum_{B_i \in \BSet} (1 -
   P(B_i))
\]

Note that the value of $\Phi_P$ is always an integer multiple of
$\eps$.  We will first show that $\gamma^{t}$ forms a non-decreasing
sequence over time, and then argue that, for any $\gamma$, the market
can stay in a \FSB{\gamma} state for a bounded number of
steps. Specifically, we will show that, if $\gamma$ does not change,
$\Phi_P$ is a non-increasing function and can stay unchanged for at
most $\gamma$ steps. Since the maximum value of $\Phi_P$ is bounded by
$\nSum$, it follows that after at most $(\gamma \nSum)/\eps$
steps, the market moves from a $\gamma$-feasible state to a
$(\gamma+1)$-feasible state (or converges).

We argue that $\gamma^t$ forms a non-decreasing sequence over time.
Since any recognized player is unmatched, if the action of an
unmatched player $Z$ results in a change in the matching, $Z$ either
matches with another unmatched player, or matches to a player that was
already matched. In either case, the total number of matched pairs
does not decrease.

Furthermore, we prove if $\gamma$ does not change, then $\Phi_{P}$ is
non-increasing. Moreover, the number of successive steps that $\Phi_P$
stay unchanged is at most $\gamma$.

To see that $\Phi_{P}$ is non-increasing, first note that $\Phi_P$ can
increase only when either a buyer decreases his bid or a seller
increases his offer.  Assume an unmatched buyer $B$ is recognized
(seller case is analogous), and the price function before his move is
$P$. To increase $\Phi_{P}$, since $B$ can only increase his bid, he
must increase an offer by overbidding and matching with a seller $S$,
resulting in the two of them submitting the same price $b$. The buyer
bid increases by $b - P(B)$ and the seller offer increases by $b -
P(S)$. Since $B$ is unmatched, by Lemma~\ref{lma:no-over-price}, $P(B)
\le P(S)$, and hence $\Phi_P$ will not increase.

We now bound the maximum number of steps for which $\Phi_P$ could
remain unchanged. A move from a buyer $B$ that does not change
$\Phi_P$ occurs only when $B$ overbids a matched seller $S$, where the
bid and the offer are equal both before and after the move. We call
this a \emph{no-change} buyer move. By analogy, a no-change seller
move can be defined.

In the remainder of the proof, we first argue that a no-change buyer
move can never be followed by a no-change seller move, and vice versa.
After that, we prove the upper bound on the number of consecutive
no-change moves to show that $\Phi_P$ will eventually decrease (by at
least $\eps$).

Assume at time $t_1$, a buyer $B_{t_1}$ made a no-change move and
matched with a seller $S_{t_1}$, who was originally paired with the
buyer $B_{t_1}'$.\footnote{An action at time time $t$ will take effect
  at the time $t+1$, and $P^t$ is the price function before any action
  is made at time $t$.} We prove that no seller can make a no-change
move at time $t_1 + 1$. The case that a seller makes a no-change move
first can be proved analogously. Suppose at time $t_1 + 1$, a seller
$S_{t_1+1}$ is recognized and decreases his offer by \eps.  Since
$B_{t_1}$ made a no-change move, we have
\begin{equation}
\label{eq:no-change}
P^{t_1}(B_{t_1}') = P^{t_1}(B_{t_1})
\end{equation}
Denote the lowest seller offer (highest buyer bid) at any time $t$ by
$s^{t}$ ($b^{t}$). By Lemma~\ref{lma:no-over-price}, $P^{t_1} (B_{t_1}) \le
P^{t_1}(S)$ for any seller $S$, hence $P^{t_1}(B_{t_1}) \le
s^{t_1}$. Moreover, since $P^{t_1}(B_{t_1}) = P^{t_1}(S_{t_1}) \ge
s^{t_1}$, we have
\begin{equation}
\label{eq:buyer-lowest-seller}
P^{t_1}(B_{t_1}) = s^{t_1}
\end{equation}
In other words, a buyer can make a no-change move, only if his bid is
equal to the lowest offer. Similarly, if $S_{t_1 + 1}$ can make a
no-change move at time $t_1 + 1$, his offer is equal to the highest
bid.  Since the highest bid at time $t_1$ ($b^{t_1}$) is at most
$s^{t_1} + \eps$ (\PeSS{1}), after $B_{t_1}$ submits a bid of $s^{t_1}
+ \eps$, he will be submitting the highest bid at time $t_1+ 1$. Hence
\begin{equation}
\label{eq:seller-price}
P^{t_1+1}(S_{t_1 + 1}) = b^{t+1} = P^{t_1+1}(B_{t_1}) = P^{t_1+1}(B'_{t_1}) + \eps
\end{equation}
Therefore, at time $t_1+1$, after $S_{t_1+1}$ decreases his offer by
$\eps$, the unmatched buyer $B_{t_1}'$ is interested in $S_{t_1+1}$.
By the unmatched first rule, $S_{t_1+1}$ will match with an unmatched
player, hence this cannot be a no-change move.

This proves that a no-change seller move can never occur after a
no-change buyer move and vice versa.
We now prove the upper bound on the number of
consecutive no-change buyer moves.

For any sequence of consecutive no-change buyer moves, if there exists
a time $t_2$ such that $s^{t_2} > s^{t_2 -1}$, for any unmatched buyer
$B$ at time $t_2$, $P^{t_2}(B) \le s^{t_2 -1} < s^{t_2}$. By
Equation~(\ref{eq:buyer-lowest-seller}), no buyer can make any more
no-change move. Moreover, since any no-change buyer move will increase
the submission of a matched seller who is submitting the lowest offer,
after at most $\gamma$ steps, the lowest offer must increase, implying
that the length of the sequence is at most $\gamma$.

To conclude, the total number of steps that the market could stay in
\FSB{\gamma} states is bounded by $(\nSum /\eps) \gamma$. As $\gamma
\le \nSum$, the total number of steps before market converges is at
most $ \nSum^3/\eps$.
\end{proof}

\subsection{General Bipartite Graphs}
\label{sec:cvg-general}

In this section, we study the convergence of markets with an arbitrary
bipartite trading graph.  Although by Theorem~\ref{lma:mec2-reach},
using Mechanism (\ref{mech:overbid}), the market can reach any
well-behaved $\eps$-stable state, when the trading graph of a market
can be an arbitrary bipartite graph, there is no guarantee that
Mechanism~({\ref{mech:overbid}}) will actually converge.

\begin{clm}
\label{clm:unstable-overbid}
In a market whose trading graph is an arbitrary bipartite
graph, Mechanism~(\ref{mech:overbid}) may never converge.
\end{clm}

\def\fig4width{1.5in}
\def\fig3width{0.24\textwidth}

\begin{figure}
\centering
  \includegraphics[width=\fig3width]{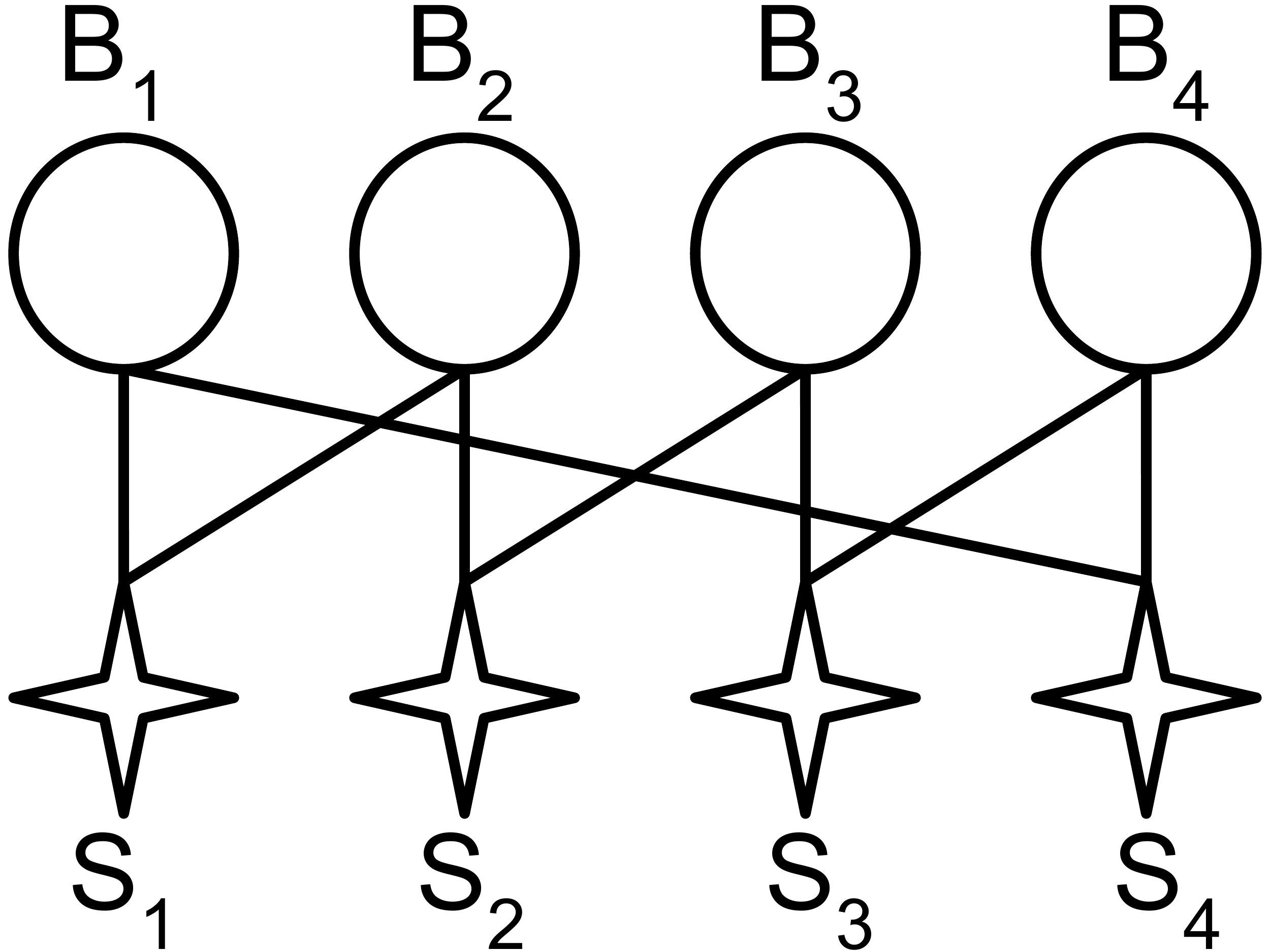}
  \includegraphics[width=\fig3width]{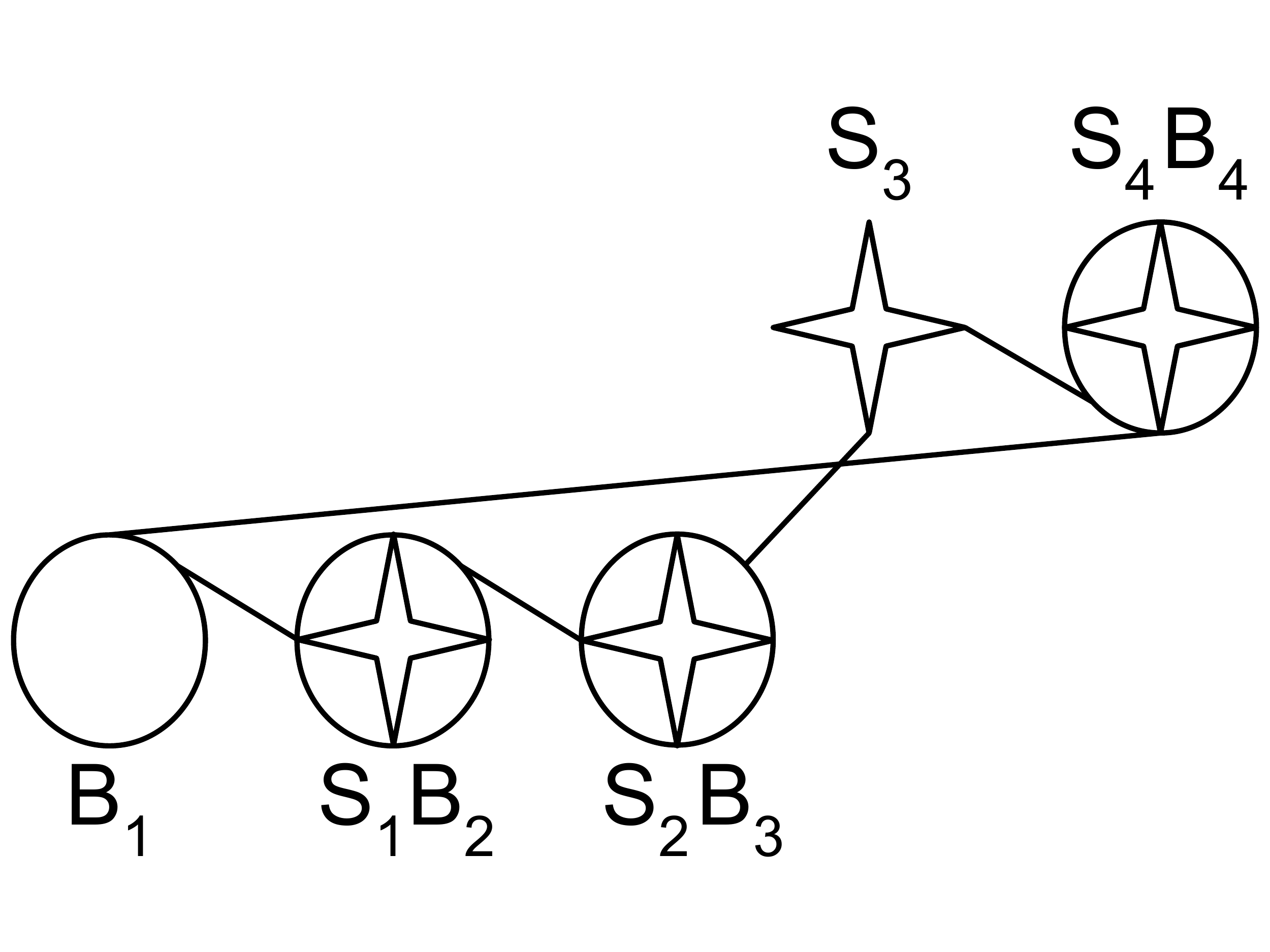}
  \includegraphics[width=\fig3width]{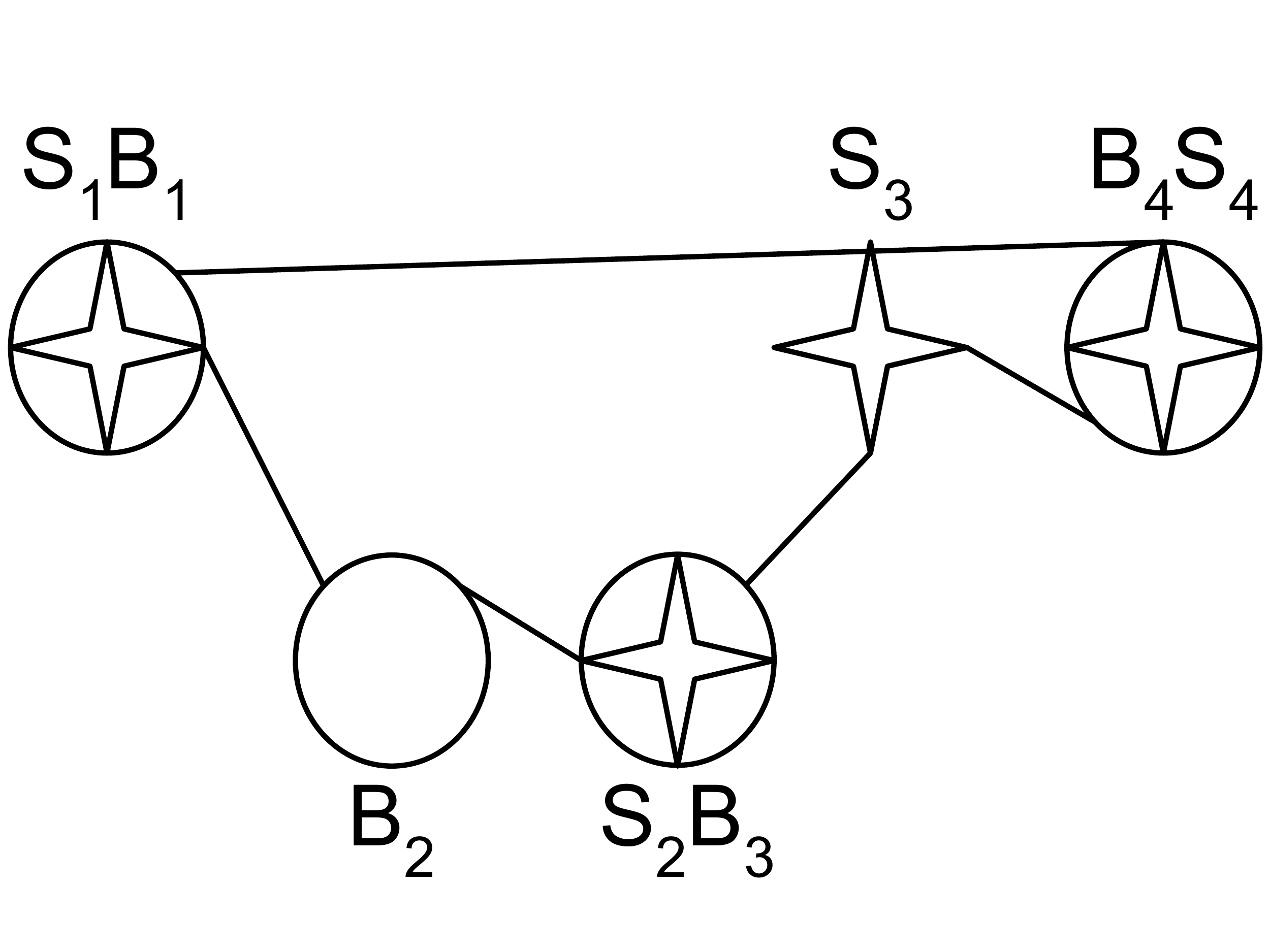}
  \includegraphics[width=\fig3width]{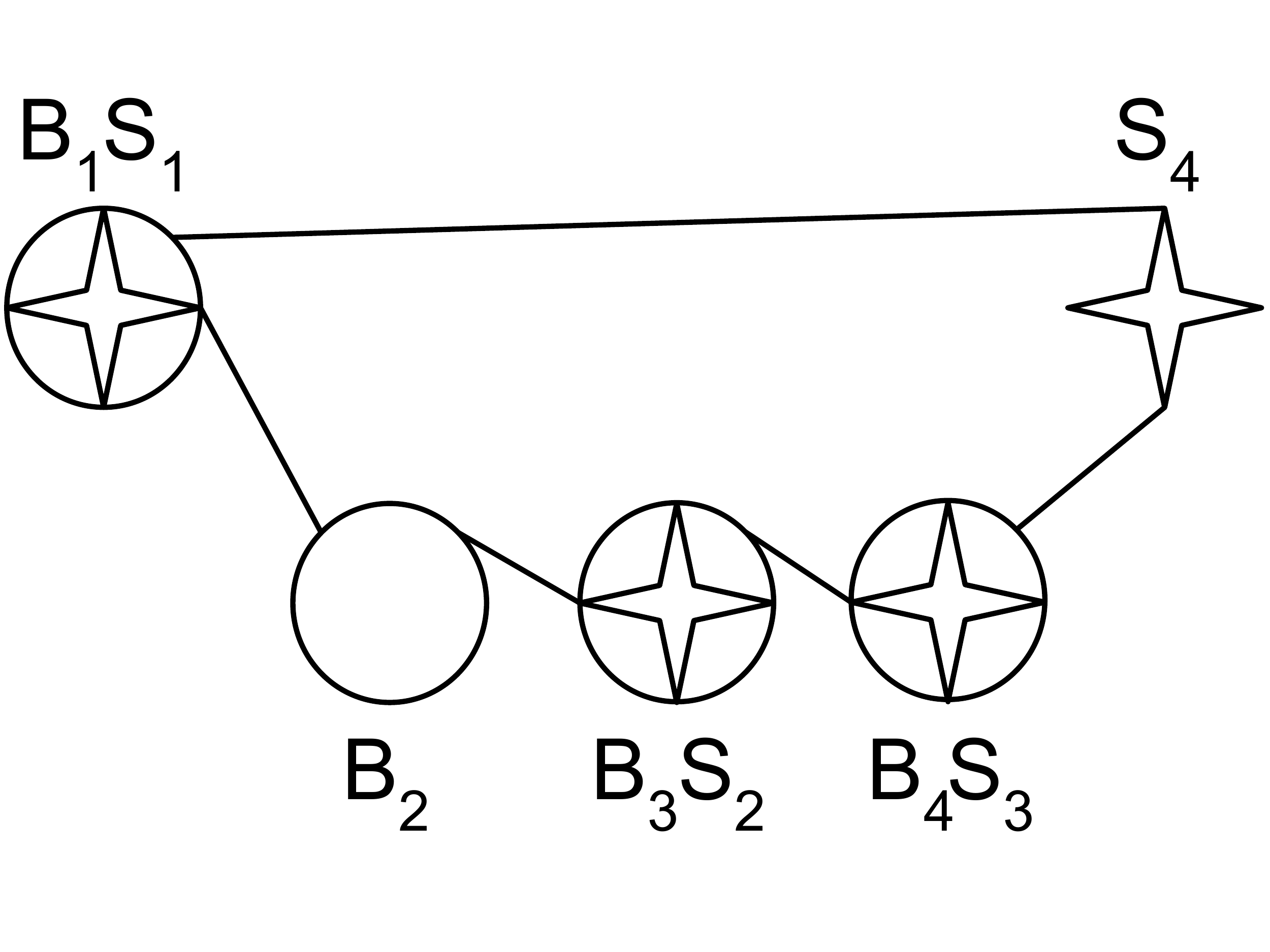}
  \caption{Unstable market with general trading graph and
    Mechanism~(\ref{mech:overbid})
     \label{fig:m2-unstable-general}}
\end{figure}

Consider the market shown in Figure~\ref{fig:m2-unstable-general}. In
this market, there are four buyers ($B_1$ to $B_4$) all with valuation
$1$ and four sellers ($S_1$ to $S_4$) all with opportunity cost
$0$. Moreover, the trading graph is a cycle of length $8$, as
illustrated by the first graph in
Figure~\ref{fig:m2-unstable-general}.  Assume at some time $t$, the
market enters the state illustrated by the second graph, where $B_1,
B_2,S_1,B_3, S_2$ are submitting $5\eps$, $S_3, B_4, S_4$ are
submitting $6\eps$, and pairs $(B_2,S_1)$, $(B_3, S_2)$ and $(B_4,
S_4)$ are matched.

At time $t+1$, since $B_1$ is unmatched, he can be recognized and
submit $6\eps$. $S_1$ is the only interested seller, hence $B_1, S_1$
will match and the offer of $S_1$ increases to $6\eps$, which leads to the
state shown in the third graph.  Similarly, at time $t+2$, since $S_3$
is unmatched, he can be recognized and submit $5\eps$. $B_4$ is the
only interested buyer, hence $B_4, S_3$ will match and bid of $B_4$
increases to $6\eps$, which leads to the state shown in the fourth
graph.

Notice that the states at time $t$ and $t+2$ are isomorphic. By
shifting the indices and repeating above two steps, the market will
never converge.

Observe that the cycle described in Claim~\ref{clm:unstable-overbid}
is caused by an adversarial coordination between the actions of
various agents.  To break this pathological coordination, we introduce
Mechanism (\ref{mech:random}) which is a natural extension of
Mechanism (\ref{mech:overbid}) that uses randomization.  We first
define this mechanism, and then prove that on any trading graph, with
high probability, the mechanism leads to convergence in a number of
steps that is polynomially bounded in the number of agents.

\begin{mechanism}\label{mech:random}
(with input parameter $\eps > 0$)
\begin{itemize}
\item \textbf{Activity Rule:} Among the unmatched buyers, any buyer
  that neither submits a new higher bid nor has a seller that is
  interested in him, is labeled as inactive. All other unmatched
  buyers are labeled as active. An active (inactive) seller is defined
  analogously. An inactive player changes his status iff some player
  on the other side matches with him.
\item \textbf{Minimum Increment:} 
  Each submitted price must be an integer multiple of \eps.
\item \textbf{Bounded Increment Rule:} In each step, a player is only
  allowed to change his price by \eps.
\item \textbf{Recognition:} Among all players who are active, one is
  recognized uniformly at random.
\item \textbf{Matching:} After a player, say a buyer $B$, is
  recognized, if $B$ does not submit a new price, then $B$ will match
  to an interested seller if one exists. If the offer of the seller is
  lower than the bid $b$, it is immediately raised to $b$. The seller
  action is defined analogously.
\item \textbf{Tie Breaking:} When choosing a player on the other side
  to match to, an unmatched player is given priority (unmatched first
  rule).
\end{itemize}
\end{mechanism}

Notice that we ask players to move cautiously through the bounded
increment rule. Players can either change the price by $\eps$ or match
with an interested seller, and always favor being active. Note that,
any move in Mechanism~(\ref{mech:overbid}) can be simulated by at most
$(1/\eps + 1)$ moves in Mechanism~(\ref{mech:random}) ($1/\eps$ for
submitting new price and $1$ for forming a match). The following is an
immediate consequence of results shown in
Section~\ref{subsec:main_mech}.

\begin{cor}\label{cor:mech2-cors}
For any market, if we use Mechanism~(\ref{mech:random}), $(i)$
starting from the zero-information state, any well-behaved
$\eps$-stable state can be reached in $\nSum (1/\eps + 1)$ steps, and
$(ii)$ beginning with a valid starting state, \PeSSs{1}{3} always
hold, and the final state is \eps-stable.
\end{cor}

We are now ready to prove our second main result, namely, for any
trading graph, with high probability, Mechanism~(\ref{mech:random})
converges to a $\eps$-stable state in a number of steps that is
polynomially bounded in the number of agents. We will utilize the
following standard fact about random walk on a line
(see~\cite{RandomizedBook}, for instance).

\begin{clm}
\label{clm:random-walk}
Consider a random walk on $\{ 0, 1, 2, ..., N \}$ such that for any $i
\in [1,N]$, the random walk transition from state $i$ to state $(i
-1)$ happens with probability $\alpha$, and for any $i \in [0, N-1]$,
the random walk transition from state $i$ to state $(i +1)$ happens
with probability $\beta$, for some $\alpha + \beta = 1$. Then starting
from any $i \in [0,N]$, with probability at least $1/2$, the random
walk either reaches the state $0$ or the state $N$, after $O(N^2)$
steps.
\end{clm}

\begin{theorem}
\label{thm:cvg-general-4}
For any market \market, if we use Mechanism~(\ref{mech:random}) with
any input parameter $\eps > 0$, and begin with a valid starting state,
the market will converge to an $\eps$-stable state after at most
$O((n^{3}/\eps^2)\log{n})$ steps with high probability.
\end{theorem}
\begin{proof}
Let \uB{t} and \uS{t} denote the number of active buyers and sellers
at time $t$, respectively, and let $u^{t} = \uB{t} + \uS{t}$.  We will
first show that \uB{t} and \uS{t} are both non-increasing functions of
time and then argue that for any $u$, with high probability the
market will remain in a state with $u$ active players for a number of
steps that is polynomially bounded in the number of players.

We first prove \uB{t} and \uS{t} are non-increasing. Note that the
only move that can make a new player active 
is one where a player, say a buyer $B$, matches to a currently matched
seller $S$.  Let $B'$ be the buyer that is currently matched to
$S$. Then at time $t+1$, the buyer $B$ moves out of the set of active
players, while the buyer $B'$ possibly joins the set of active
players.  Thus the number of active players remains unchanged.  A
similar argument applies to case when a seller is recognized and
matches to a currently matched buyer.

In the remainder of the proof, we first show that if there exists an
adjacent buyer-seller pair such that both players are unmatched and
the buyer bid is not below the seller offer (we call such a pair to be
an \emph{active pair}), then after $O(n \log{n})$ steps, with
probability $1 - O(1/n^2)$, $u^t$ will decrease.  Next, in the absence
of active pairs, we argue that either $u^t$ decreases or an active
pair appears in the market after $O((n/\eps)^2\log{n})$
steps, with probability $1 - O(1/n^2)$. Note that if an active pair
appears, by the same argument, after $O(n \log{n})$ more steps, $u^t$
will decrease with high probability.  Since $u^t \le n$, we can
conclude that the market converges in
$O((n^{3}/\eps^2)\log{n})$ steps with high probability.

We first prove that, the existence of an active pair will lead to
decrement of $u^{t}$. For any active pair $(B, S)$.  By the unmatched
first rule, recognizing either $B$ or $S$ will increase the number of
matches. Recognizing any other player who makes a move to match with
$B$ or $S$, will also increase the number of matches (note that only
unmatched players will be recognized). Both cases decrease $u^t$ by
$2$. In other words, as long as $u^t$ does not decrease, $(B,S)$ will
remain to be an active pair. Assume at time $t_1$, there is an active
pair $(B,S)$. Let $Y_{t}$ be the random variable which is $1$ iff $B$
or $S$ is recognized at time $t$. Hence, for any $t \ge t_1$ where
$u^{t} = u^{t_1}$,
\[ Pr(Y_{t} = 1) = \frac{2}{u^{t_1}} \ge \frac{2}{\nSum}\]
It follows that after $\nSum$ steps from $t_1$ the probability that
none of $B$ or $S$ is chosen is less than $1/2$. Therefore, after
$2\nSum\log{\nSum}$ steps, with probability $1-(1/\nSum)^2$, either
one of $B$ and $S$ has been recognized or $u^{t}$ has already
decreased. In either case, $u^t$ decreases.

Next, in the absence of active pairs, we prove that after a bounded
number of steps, either $u^t$ decreases or an active pair
appears. Consider the following potential function
\[\Phi = \sum_{B \in \BSet} P(B) + \sum_{S \in \SSet} P(S)\]
If there is no active pairs, by the design of the
Mechanism~(\ref{mech:random}), when recognized, any buyer will
increase $\Phi$ by $\eps$ and any seller will decrease $\Phi$ by
$\eps$.  Thus, at any time $t$ with no active pairs, the probability
that $\Phi$ increases by $\eps$ is $P^{t}_{\BSet} = \uB{t}/u^t$, and
the probability that $\Phi$ decreases by $\eps$ is $P^{t}_{\SSet} =
\uS{t}/u^t$ (note that $P^{t}_{\BSet}$ or $P^{t}_{\SSet}$ might be
$0)$.

In the following, we will use Claim~\ref{clm:random-walk} to prove
that as long as $u^t$ does not change and no active pair appears,
after a bounded number of steps, with high probability, $\Phi$ will
have reached its upper or lower bound.  If $\Phi$ reaches its upper
bound then all buyers must be submitting their true valuations and all
sellers must be submitting $1$.  Thus every unmatched buyer is
inactive and $u^t$ must have decreased. A similar situation also happens
when $\Phi$ reaches its lower bound.

To use Claim~\ref{clm:random-walk}, see that if $u^t$ does not change and no active pair appears,
$P_\BSet^t$ and $P_\SSet^t$ will also remain unchanged. During this time
period, we can denote
the probability of $\Phi$ increases by $\eps$ as $P_\BSet$ and the
probability of $\Phi$ decreases by $\eps$ as $P_\SSet$. Let $\alpha =
P_\SSet$, $\beta = P_\BSet$, and nodes be $\set{0, \eps, 2\eps,\dots,
  n}$ (hence $N = n/\eps$). Thus this is a random walk, and by
Claim~\ref{clm:random-walk}, after $O((n/\eps)^2)$ steps, the
probability of $\Phi$ reaches its upper bound $n$ or lower bound $0$
is at least $1/2$. Therefore, after $O((n/\eps)^2 \log{n})$ steps,
with probability at least $1-O(1/n^2)$, $\Phi$ reaches $0$ or $n$.

To conclude, after $O((n/\eps)^2\log{n})$ steps, $u^t$
will decrease with probability at least $1 - O(1/n^{2})$. As $u^{t}\le
n$, by union bound, the market will converge after
$O((n^{3}/\eps^2)\log{n})$ steps with probability
$1-O(1/n)$. 

\end{proof}

\section{Conclusions}
\label{sec:conclusion}

We resolved the second part of Friedman's conjecture by designing a
mechanism which simulates the DOA and proving that this mechanism
always converges to a Walrasian equilibrium in polynomially many
steps.  Our mechanism captures four key properties of the DOA: agents
on either side can make actions; agents only have limited information;
agents can choose \emph{any} better response (as opposed to the best
response); and the submissions are recognized in an arbitrary order.
An important aspect of our result is that, unlike previous models,
\emph{every} Walrasian equilibrium is reachable by some sequence of
better responses.

For markets where only a restricted set of buyer-seller pairs are able
to trade, we show that the DOA may never converge. However, if
submissions are recognized randomly, and players only change their
bids and offers by a small fixed amount, convergence is guaranteed.
It is unclear that the latter condition is inherently necessary, and
perhaps a convergence result can be established for a relaxed notion
of bid and offer changes where players can make possibly large
adjustments as long as they are consistent with the increment rule.

\bibliographystyle{plain}
\bibliography{ref}

\clearpage
\appendix
\section{Omitted Details of Section~\ref{sec:ss-sw}}
\subsection{Proof of Theorem~\ref{thm:stable-pd}}\label{sec:ss-msw}

\begin{proof}
To see the forward direction, assume \MSPP\ is a stable state.  We
first verify that $x^\Pi$ and $y^P$ are indeed feasible solutions.
$x^\Pi$ is clearly feasible since it is characteristic function of a
valid matching. $y^P$ preserves non-negativity constraint of dual,
since no player could submit a price exceeding his valuation in $P$.
Moreover, we can write $y^P_B + y^P_S = \val{B} - P(B) + P(S) -
\val{S} = w_{B,S} + P(S) - P(B)$. By \PSS{1}, $P(S) \ge P(B)$, hence
$y^P_B + y^P_S \ge w_{B,S}$, preserving the dual constraint and
implying that $y^P$ is also feasible.

To prove optimality of $x^\Pi$ and $y^P$, using weak duality, we only
need to verify that value of primal is equal to value of dual.

\begin{align}
\sum_{ B\in \BSet} y_{B} + \sum_{S\in \SSet} y_{S} &= \sum_{\match{B,S}
  \in \Pi} (\val{B} - P(B) + P(S) - \val{S}) + \sum_{Z \notin \Pi}
y^P_Z\\ &= \sum_{\match{B,S}\in \Pi} (\val{B} - \val{S}) =
\sum_{\edge{B,S}\in E} w_{B,S} \times x^\Pi_{B,S}
\end{align}
$(1)$ to $(2)$ uses \PSSs{1}{2}, $P(B) = P(S)$ for $\match{B,S} \in
\Pi$, and $P(Z) = \val{Z}$ for $Z \notin \Pi$. Thus $x^\Pi$ and $y^P$
are optimal solutions of primal and dual, respectively.

For the reverse direction, assume $(x^\Pi, y^P)$ is a pair of optimal
primal and dual solutions . Since $y^P$ is a feasible solution, as we
just stated, $y^P_B + y^P_S \ge w_{B,S}$ will give us $P(S) \ge P(B)$,
thus \PSS{1} holds. For \PSSs{2}{3}, since
\begin{align*}
\sum_{ B\in \BSet} y_{B} + \sum_{S\in \SSet} y_{S} &=
\sum_{\edge{B,S}\in E} w_{B,S} \times x^\Pi_{B,S} \\ \sum_{\match{B,S} \in \Pi}
(\val{B} - P(B) + P(S) - \val{S}) + \sum_{Z \notin \Pi} y^P_Z &=
\sum_{\match{B,S}\in \Pi} \val{B} - \val{S} \\ \sum_{\match{B,S} \in \Pi}
(P(S) - P(B)) + \sum_{Z \notin \Pi} y^P_Z &= 0
\end{align*}
Since $P(S) \ge P(B)$ and also $y^P$ is a non-negative vector, both
terms in the last expression must be zero, which implies that
\PSSs{2}{3} also hold.  Therefore \MSPP\ is a stable state.  
\end{proof}

\subsection{Proof of Theorem~\ref{thm:MSW-eSS}}\label{sec:proof-esw}

To simplify the notation, we treat an agent who is unmatched as being
matched with themselves. To this end, for each buyer we introduce a
\emph{dummy} seller with an opportunity cost equal to his valuation,
similarly for each seller. An agent matched with their dummy
counterpart is interpreted as being unmatched. We denote the dummy
seller of buyer $B$ as \shadowB~and the dummy buyer of seller $S$ as
\shadowS.

\begin{proof}
We define the following $\eps$-primal and $\eps$-dual pair.

\begin{align*}
&\max{\sum_{\edge{B,S}\in E}  (w_{B,S} - \eps) x_{B,S}}
  & &\min{\sum_{ B\in \BSet}  y_{B} + \sum_{S\in \SSet}  y_{S}}\\
\text{s.t. } &\forall B^* \in \BSet, \sum_{\edge{B^*,S} \in E}x_{B^*,S} \le 1
 & \text{s.t. } &\forall \edge{B,S} \in E, y_B + y_S \ge (w_{B,S} - \eps)\\
&\forall S^* \in \SSet, \sum_{\edge{B,S^*} \in E}x_{B,S^*} \le 1
& & y_{B}, y_{S} \ge 0\\
& x_{B,S} \ge 0
\end{align*}

Given a $\eps$-stable state \MSPP, since \PeSS{1} is equivalent to
$\val{B} - P(B) + P(S) - \val{S} \ge w_{B,S} - \eps$, $y^P$ is a
feasible solution of $\eps$-dual (non-negativity constrains hold since
$P$ is valid price function). By \PeSSs{2}{3},
\begin{alignat*}{2}
\sum_{ B\in \BSet}  y^P_{B} + \sum_{S\in \SSet}  y^P_{S} &= \sum_{\match{B,S} \in \Pi} (\val{B} - P(B) + P(S)- \val{S}) + \sum_{Z  \notin \Pi} y^P_Z \\ 
&= \sum_{\match{B,S} \in \Pi} (\val{B} - \val{S}) = \SWV_{\Pi}
\end{alignat*}

On the other hand, if we take a matching $\Pi^*$ that maximizes SW, and plug
$x^{\Pi^*}$ into the $\eps$-primal, we have 
\[\sum_{\edge{B,S}\in E}
(w_{B,S} - \eps) x_{B,S} = \sum_{\edge{B,S}\in E}w_{B,S}x_{B,S} -
\sum_{\edge{B,S}\in E} \eps x_{B,S} \ge \SWV_{\Pi^*} - \nMin\eps\] 
The last inequality comes from the fact
that \nMin\ is an upper bound on the number of possible pairs (i.e., number
of possible $1$'s in $x_{B,S}$) for any matching. By the weak
duality, any value of $\eps$-primal is less than or equal to any value
of $\eps$-dual, thus $\SWV_\Pi \ge \SWV_{\Pi^*} - \nMin\eps$.

We now proceed to prove the condition for an \eps-stable state to
maximize SW.  Fix a matching $\Pi^*$ that maximizes SW. Construct graph
$G' (V',E')$ with $V' = \BSet \cup \SSet$ and $E' =
\setst{\edge{B,S}}{\match{B,S} \in \Pi \vee \match{B,S} \in \Pi^*}$. As any
player can be matched with at most one other player in each matching, the
degree of each node in $G'$ is at most two. Consequently, the connected
components of $G'$ could only be cycles or paths. Note that such
cycles and paths are formed by the different pairs of the two
matchings. We now prove that for any of those cycles or paths, the local
SW of the two matchings are the same.

For any cycle $B_0, S_0, B_1, S_1, \dots, B_k, S_k, B_0$, pair
$\match{B_i,S_i}$ belongs to one matching while pair $\match{ B_{i+1},
  S_i}$ belongs to the other one. If we only consider these players,
every buyer gets an item and every seller sells the item, thus the SW
of both matchings are the same.

For any path $Z_0, Z_1, Z_2, Z_3, \dots, Z_k$, wlog, we can assume
$\val{Z_0} \ge \val{Z_k}$. If $Z_0$ is a seller, add his dummy buyer
to the left of the path. If $Z_k$ is a buyer, add his dummy seller to
the right of the path. Therefore, the path starts with a buyer and
ends with a seller. We can denote the path as $B_0, S_0, B_1, S_1,
\dots, B_k, S_k$.

For the same reason as cycle case, the players in the middle
contributes same amount of SW to both matchings, thus the difference
of SW is $\val{B_0} - \val{S_k}$. Since $\Pi^*$ is a matching that
maximizes SW, it must be the case that $\match{B_{i}, S_{i}} \in
\Pi^*$ and $\match{B_{i+1}, S_i} \in \Pi$.

If the difference of SW is $0$, then we are done.  Suppose not, then
$\val{B_0} - \val{S_k} \ge \delta$.
By \PeSSs{1}{3},
\[P(B_{i+1}) = P(S_i) \ge P(B_i)  - \eps \Rightarrow P(B_0) \le P(B_k) + k\eps\]
We now have
\begin{align}
\val{B_0} - \val{S_k} &= P(B_0) - P(S_k) \\
&\le P(B_k) + k\eps - P(S_k) \le (k+1)\eps \le \nMin \eps < \delta \label{eq:delta}
\end{align}
where $(3)$ is because both $B_0$ and $S_k$ are matched in $\Pi^*$ but
not in $\Pi$, implying that their submitted prices are equal to their
valuation.

Thus on one side we have $val(B_0) - val(S_k) \geq \delta$, and from
the other side by inequality $(\ref{eq:delta})$ above we have
$val(B_0) - val(S_k) < \delta$, a contradiction.  It concludes that
all such cycles and paths generate the same SW for both matchings and
thus $\Pi$ also maximizes SW. 
\end{proof}

\end{document}